\documentclass{article}

\usepackage{amsmath}
\usepackage{amssymb}
\usepackage[linesnumbered,ruled,vlined]{algorithm2e}
\usepackage{algpseudocode}
\algnewcommand{\algorithmicgoto}{\textbf{go to}}%
\algnewcommand{\Goto}[1]{\algorithmicgoto~\ref{#1}}%
\DeclareMathOperator*{\argmin}{\arg\!\min}
\usepackage{multirow}
\usepackage[table]{xcolor}
\usepackage{graphicx}

\usepackage{amsthm}
\usepackage{url}
\usepackage{epsfig}

\newtheorem{theorem}{Theorem}
\newtheorem{lemma}{Lemma}
\newtheorem{definition}{Definition}

\usepackage[plainpages=false,pdfpagelabels]{hyperref}

\begin{document}
\title{Common Due-Date Problem: \\Exact Polynomial Algorithms for a Given Job Sequence}

\author{Abhishek Awasthi$^*$, J\"org L\"assig$^*$ and Oliver Kramer$^\dag$\\ \\
$^*$Department of Computer Science\\
University of Applied Sciences Zittau/G\"orlitz \\ G\"orlitz, Germany\\
\{abhishek.awasthi, joerg.laessig\}@hszg.de
\and
$^\dag$Department of Computing Science\\
Carl von Ossietzky University of Oldenburg\\
Oldenburg, Germany \\
oliver.kramer@uni-oldenburg.de
}

\maketitle
\begin{abstract}
This paper considers the problem of scheduling jobs on single and parallel machines where all the jobs possess different processing times but a common due date. There is a penalty involved with each job if it is processed earlier or later than the due date. The objective of the problem is to find the assignment of jobs to machines, the processing sequence of jobs and the time at which they are processed, which minimizes the total penalty incurred due to tardiness or earliness of the jobs. This work presents exact polynomial algorithms for optimizing a given job sequence for single and parallel machines with the run-time complexities of $O(n \log n)$ and $O(mn^2 \log n)$ respectively, where $n$ is the number of jobs and $m$ the number of machines. The algorithms take a sequence consisting of all the jobs $(J_i, i=1,2,\dots,n)$ as input and distribute the jobs to machines (for $m>1$) along with their best completion times so as to get the least possible total penalty for this sequence. We prove the optimality for the single machine case and the runtime complexities of both. Henceforth, we present the results for the benchmark instances and compare with previous work for single and parallel machine cases, up to $200$ jobs.
\end{abstract}

\section{Introduction}
The Common Due-Date (CDD) scheduling problem involves sequencing and scheduling of jobs over machine(s) against a common due-date. Each job possesses a processing time and different penalties per unit time in case the job is completed before or later than the due-date. The objective of the problem is usually to schedule the jobs so as to get the least total penalty due to earliness or tardiness of all the jobs.  In practice, a common due date occurs in almost any manufacturing industry and its goal is to reduce both, the earliness and tardiness of the goods produced. Earliness of the produced goods is not desired because it requires the maintenance of some stocks leading to some expenses to the industry for storage cost, tied-up capital with no cash flow etc. On the other hand, a tardy job leads to customer dissatisfaction.

When scheduling on a single machine against a common due date, one job at most can be completed exactly at the due date. Hence, some of the jobs might get completed earlier than the common due-date ($d$), while other jobs finish late. Generally speaking, there are two classes of common due-date problems which have proven to be NP-hard, namely the restrictive and non-restrictive CDD problem. In this work we consider the restrictive case of the problem where the common due date is less than the sum of the processing times of all the jobs and each job possess different earliness/tardiness penalties. This problem has proven to be the most difficult problem in this area of research~\cite{biskup}.

\section{Related Work}
Common due-date problems have been studied extensively during the last $30$ years with several variants and special cases~\cite{seidmann,kanet}. In $1981$, Kanet presented an O($n \log n $) algorithm for minimizing the total absolute deviation of the completion of jobs from the due date for the single machine~\cite{kanet}. Panwalkar~\emph{et al.} considered the problem of common due-date assignment to minimize the total penalty for one machine~\cite{panwalkar}. The objective of the problem was to determine the optimum value for the due-date and the optimal job sequence to minimize the penalty function, where the penalty function also depends on the due-date along with earliness and tardiness. An algorithm of O($n \log n$) complexity was presented but the problem considered by them consisted of symmetric costs for all the jobs~\cite{seidmann,panwalkar}. Cheng again considered the same problem with slight variations and presented a linear programming formulation~\cite{cheng1}. In $1991$ Cheng~\emph{et al.} made a very good study on the common due date problem and presented some useful properties for the general case~\cite{cheng}.

A pseudo-polynomial algorithm of O($n^2d$) complexity was presented by Hoogeveen~\emph{et al.} for the restrictive case with one machine when the earliness and tardiness penalty weights are symmetric for all the jobs~\cite{hoogeveen}. In $1991$ Hall studied the un-weighted earliness and tardiness problem and presented a dynamic programming algorithm~\cite{hall}. Besides these earlier works, there have been some research on heuristic algorithms for the general common due date problem with asymmetric penalty costs. James presented a tabu search algorithm for the general case of the problem in $1997$~\cite{james}. More recently Feldmann~\emph{et al.} approached the problem using metaheuristic algorithms namely simulated annealing (SA) and threshold accepting and presented the results for benchmark instances up to $1000$ jobs on a single machine~\cite{biskup,feldmann}. Another variant of the problem was studied by Toksari~\emph{et al.} where they considered the common due date problem on parallel machines under the effects of time dependence and deterioration~\cite{toksari}. Ronconi~\emph{et al.} proposed a branch and bound algorithm for the general case of the CDD and gave optimal results for small benchmark instances~\cite{ronconi}. Rebai~\emph{et al.} proposed metaheuristic and exact approaches for the common due date problem to schedule preventive maintenance tasks~\cite{rebai}.

In this paper we consider both single and parallel machine cases for the CDD problem with asymmetric penalties. We present polynomial exact algorithms to optimize a given job sequence on single and parallel machines. We present our results and compare the single machine results with the results provided in the OR-library~\cite{or}. Besides, we also present new results for the benchmark instances for parallel machines up to $200$ jobs with $k=1$ and the restrictive factor ($h$) of $0.4$ and $0.8$, provided by Biskup~\emph{et al.} in~\cite{biskup}.

\section{Problem Formulation}\label{probform}
In this section we give the mathematical notation of the common due date problem based on~\cite{biskup}. We also define some new parameters which are later used in the presented algorithm in the next section.

\noindent
Let, \\
$n$  =  total number of jobs \\
$m$ = total number of machines \\
$n_{j}$ = number of jobs processed by machine $j$ $( j = 1,2,\dots,m)$  \\
$M_{j}$ = time at which machine $j$ finished its previous job \\
$W_{j}^{k}$ = $k^{th}$ job processed by machine $j$ \\
$P_{i}$ = processing time of job $i$ $( i = 1,2,\dots,n)$  \\
$C_{i}$ = completion time of job $i$ $( i = 1,2,\dots,n)$  \\
$D$ = the common due date \\
$\alpha_{i}$ = the penalty cost per unit time for job $i$ for being early\\ 
$\beta_{i}$ = the penalty cost per unit time for job $i$ for being tardy\\ 
$E_{i}$ = earliness of job $i$, $E_i = \max\{0,D-C_{i}\}$ ($i = 1,2,\dots,n$) \\
$T_{i}$ = tardiness of job $i$, $T_{i}=\max\{0,C_{i}-D\}$ ($i = 1,2,\dots,n$)\;.

The cost corresponding to job $i$ is then expressed as $\alpha_{i}E_{i}+\beta_{i}T_{i}$. If job $i$ is completed at the due date then both $E_{i}$ and $T_{i}$ are equal to zero and the cost assigned to it is zero. When job $i$ does not complete at the due date, either $E_{i}$ or $T_{i}$ is non-zero and there is a strictly positive cost incurred. The objective function of the problem can now be defined as
\begin{equation}\label{ob}
\min \sum\limits_{i=1}^{n} (\alpha_{i}E_{i}+\beta_{i}T_{i}) \;.
\end{equation}

\section{The Exact Algorithm}
Let $J$ be the input job sequence where $J_i$ is the $i$th job in the sequence $J$. Note that without loss of any generality we can assume $J_i=i$, since we can rank the jobs for any sequence as per their order of processing. The algorithm takes the job sequence $J$ as the input and returns the optimal value for Equation~\eqref{ob}. There are three requirements for the optimal solution: allotment of jobs to specific machines, the order of processing of jobs in every machine and the completion times for all the jobs. 

\begin{figure}[bht]
\centering
\includegraphics[width=10cm,height=3cm]{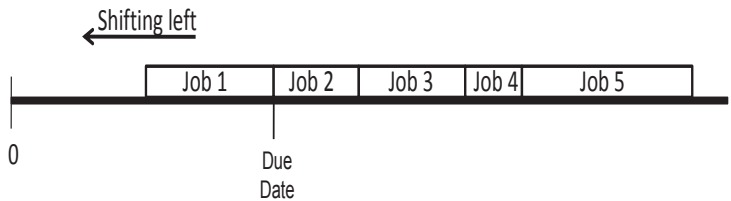}
\caption{Left Shift (decrease in completion times) of all the jobs towards decreasing total tardiness for a sequence with $5$ jobs. Each reduction is done by the minimum of the processing time of the job which is starting at the due date and the maximum possible left shift for the first job.}
\label{fig1}
\end{figure}

Cheng~\emph{et al.}~\cite{cheng} proved some properties of the optimal solution for the restrictive common due date case, one among which is that the optimal solution for the restrictive common due date case  with general penalties exists with no idle times between consecutive jobs. This property implies that all the jobs in any sequence are processed right after another. Using this property, our algorithm assigns the completion times to all the jobs such that the first job is finished at $\max\{P_{1},D\}$ and the rest of the jobs follow without any idle time in order to obtain an initial solution which is then improved incrementally. It is quite apparent that a better solution for this sequence can be found only by reducing the completion times of all the jobs, {\em i.e.\ \/}shifting all the jobs towards decreasing total tardiness penalty as shown in Figure~\ref{fig1} with $5$ jobs. Shifting all the jobs to the right will only increase the total tardiness. 

Hence, we assign the jobs in $J$ to machines such that they incur least weighted tardiness initially and assign the initial completion times to the jobs in each machine as stated in Equation~\eqref{initial}.

\begin{definition}
We define $\lambda$ as the machine which has the least value for $M_{j}$, $j=1,2,\dots,m$, i.e.
\begin{equation*}
\lambda = \argmin_{j=1,2,\dots,m} M_j \;.\\
\end{equation*}
\end{definition}
\begin{algorithm}
\DontPrintSemicolon
$M_{j} \gets 0$ $\forall j=1,2,\dots,m$\;
$n_j \gets 0$ $\forall j=1,2,\dots,m$\;
$i \gets 0$\;
\For {$j \gets 1$ \normalfont\textbf{to} $m$}{
$i \gets i+1$\;
$n_j \gets n_j+1$\;
$W_{j}^{1} \gets i$\;
$M_{j} \gets \max\{P_{i},D\}$\;
}

\For {$i \gets m+1$ \normalfont\textbf{to} $n$}{
\textbf{Compute} $\lambda$\;
$n_{\lambda} \gets n_{\lambda}+1$\;
$W_{\lambda}^{n_{\lambda}} \gets i$\;
$M_{\lambda} \gets M_{\lambda}+P_{i}$\;
}
\caption{Allotment of Jobs to Machines}
\label{alotmac}
\end{algorithm}

Algorithm~\ref{alotmac} assigns the first $m$ jobs to each machine respectively such that they all finish processing at the due date or after their processing time, which ever is higher. For the remaining jobs, we assign a machine $\lambda$ to job $i$ since it offers the least possible tardiness. Likewise each job is assigned at a specific machine such that the tardiness for all the jobs is the least for the given job sequence. The job sequence is maintained in the sense that for any two jobs $i$ and $j$ such that job $j$ follows $i$; the Algorithm~\ref{alotmac} will either maintain this sequence or assign the same starting times at different machines to both the jobs. Finally, Algorithm~\ref{alotmac} will give us the number of jobs $(n_{j})$ to be processed by any machine $j$ and the sequence of jobs in each machine, $W_{j}^{k}$. This is the best assignment of jobs at machines for the given sequence. Note that the sequence of jobs is still maintained here, since Algorithm~\ref{alotmac} ensures that any job $i$ is not processed after a job $i+1$. Once we have the jobs assigned to particular machines, the problem then converts to $m$ single machine problems, since all the machines are independent.

We now present the ideas and the algorithm for solving the single machine case for a given job sequence. From here on we assume that there are $n$ jobs to be processed by a machine and all the parameters stated at the beginning of Section~\ref{probform} represent the same meaning.

\begin{lemma}\label{lemma}
If the initial assignment of the completion times of the jobs, for a given sequence $J$ is done according to $C_{i}$ where,
\begin{equation}\label{initial}
C_{i} = 
\begin{cases}
\max\{P_{1},D\} & \mbox{if } i=1 \\ 
C_{i-1}+P_{i} & \mbox{if } 2 \leq i \leq n \;,
\end{cases}
\end{equation} then the optimal solution for this sequence can be obtained only by reducing the completion times of all the jobs or leaving them unchanged.
\end{lemma}

\begin{proof}
We prove the above lemma by considering two cases of Equation~\eqref{initial}.
\noindent
\textbf{Case 1:} $D > P_{1}$ \\
In this case Equation~\eqref{initial} will ensure that the first job is completed at the due-date and the following jobs are processed consecutively without any idle time. Moreover, with this assignment all the jobs will be tardy except for the first job which will be completed at the due date. The total penalty (say, $PN$) will be $\sum_{\substack{i=1}}^{n} (\beta_{i} \cdot T_{i})$, where $T_{i}=C_{i}-D$, $i=1,2,\dots,n$. Now if we increase the completion time of the first job by $x$ units then the new completion times $C_{i}^{\prime}$ for the jobs will be $C_{i} + x$ $\forall$ $i,(i=1,2,\dots,n)$ and the new total penalty $PN^{\prime}$ will be $\sum_{\substack{i=1}}^{n} (\beta_{i} \cdot T_{i}^{'})$, where $T_{i}^{'}=T_{i}+x$ $(i=1,2,\dots,n)$. Clearly, we have $PN^{\prime} > PN$ which proves that an increase in the completion times cannot fetch optimality which in turn proves that optimality can be achieved only by reducing the completion times or leaving them unchanged from Equation~\eqref{initial}.
\noindent
\textbf{Case 2:} $D \leq P_{1}$ \\
If the processing time of the first job in any given sequence is more than the due-date then all the jobs will be tardy including the first job as $P_1>D$. Since all the jobs are already tardy, a right shift ({\em i.e.\ \/}increasing the completion times) of the jobs will only increase the total penalty hence worsening the solution. Moreover, a left shift ({\em i.e.\ \/}reducing the completion times) of the jobs is not possible either, because $C_{1}=P_{1}$, which means that the first job will start at time $0$. Hence, in such a case Equation~\eqref{initial} is the optimal solution. In the rest of the paper we avoid this simple case and assume that for any given sequence the processing time of the first job is less than the due-date.
\end{proof}

Before stating the algorithm we first introduce some new parameters, definitions and theorems which are useful for the description of the algorithm. We first define $DT_{i} = C_{i}-D$, $i=1,2,\dots,n$ and $ES=C_{1}-P_{1}$. It is clear that $DT_{i}$ is the algebraic deviation of the completion time of job $i$  from the due date and $ES$ is the maximum possible shift (reduction of completion time) for the first job.

\begin{lemma}\label{lemma2}
Once $C_i$ for each job in a sequence is assigned according to Lemma~\ref{lemma}, a reduction of the completion times is possible only if $ES>0$.
\end{lemma}

\begin{proof}
Lemma~\ref{lemma} proves that only a reduction of the completion times can improve the solution once the initialization is made as per Equation~\eqref{initial}. Besides there is no idle time between any jobs, hence an improvement can be achieved only if $ES>0$, in which case all the jobs will be shifted left by equal amount.
\end{proof}

\begin{definition}\label{pl}
$PL$ is a vector of length $n$ and any element of $PL$ ($PL_{i}$) is the penalty possessed by job $i$. We define $PL$, as
\begin{equation}
PL_{i} =
\begin{cases}
	-\alpha_{i}, & \mbox{if }DT_{i} \leq 0 \\
	\beta_{i}, & \mbox{if }DT_{i} > 0 \;.
\end{cases}
\end{equation}
\end{definition}

With the above definition we can now express the objective function stated by Equation~\eqref{ob} as $\min(Sol)$, where $Sol$:
\begin{equation}\label{obn}
Sol = \sum\limits_{i=1}^{n} (DT_{i} \cdot PL_{i}) \;.
\end{equation}

\begin{algorithm}
\DontPrintSemicolon
\textbf{Initialize} $C_{i}$ $\forall$ $i$ (Equation~\ref{initial})\;
\textbf{Compute} $PL,DT,ES$\;
$Sol \gets \sum\limits_{i=1}^{n} (DT_{i} \cdot PL_{i})$\;
$j \gets 2$\;
\While{$(j < n+1)$}{
$C_{i} \gets C_{i}-\min\{ES,DT_{j}\}$, $\forall$ $i$\;
\textbf{Update} $PL,DT,ES$\;
$V_{j} \gets \sum\limits_{i=1}^{n} (DT_{i} \cdot PL_{i})$\;
\lIf {$(V_{j} < Sol)$}
	{ $Sol \gets V_{j}$}
\lElse{
	\Goto {return}
}
$j \gets j+1$\;
}
\Return $Sol$ \label{return} \;
\caption{Exact Algorithm for Single Machine}
\label{main}
\end{algorithm}

For the parallel machine case we just need to first assign all the jobs to machines using Algorithm~\ref{alotmac} and then implement Algorithm~\ref{main} for each machine $j$ with $n_j$ jobs to optimize the jobs sequence in every machine, as stated in Algorithm~\ref{parallel}.

\begin{algorithm}
\DontPrintSemicolon
\textbf{Algorithm~\ref{alotmac}}\;
\For {each machine}{
\textbf{Algorithm~\ref{main}}\;
}
\caption{Exact Algorithm: Parallel Machine}
\label{parallel}
\end{algorithm}

\section{Proof of Optimality}

\begin{theorem}\label{theorem}
Algorithm~\ref{main} finds the optimal solution for a single machine common due date problem, for a given job sequence.
\end{theorem}

\begin{proof}
The initialization of the completion times for a sequence $P$ is done according to Lemma~\ref{lemma}.  It is evident from Equation~\eqref{initial} that the deviation from the due date ($DT_{i}$) is zero for the first job and greater than zero for all the following jobs. Besides, $DT_i < DT_{i+1}$ for $i=1,2,3,\dots,n-1$, since $C_{i}<C_{i+1}$ from Equation~\eqref{initial} and $DT_{i}$ is defined as $DT_{i}=C_{i}-D$. By Lemma~\ref{lemma} the optimal solution for this sequence can be achieved only by reducing the completion times of all the jobs simultaneously or leaving the completion times unchanged.

The total penalty after the initialization is $PN=\sum_{\substack{i=1}}^{n} (\beta_{i} \cdot T_{i})$ since none of the jobs are completed before the due date. According to Algorithm~\ref{main} the completion times of all the jobs is reduced by $\min\{ES,DT_{j}\}$ at any iteration. Since $DT_1=0$, there will be no loss or gain for $j=1$. After any iteration of the $while$ loop in line $5$, we decrease the total weighted tardiness but gain some weighted earliness penalty for some jobs. A reduction of the completion times by $\min\{ES,DT_{j}\}$ is the best non-greedy reduction. Let $\min\{ES,DT_{j}\}>0$ and $t$ be a number between $0$ and $\min\{ES,DT_{j}\}$. Then reducing the completion times by $t$ will increase the number of early jobs by one and reduce the number of tardy jobs by one. With this operation; if there is an improvement to the overall solution then a reduction by $\min\{ES,DT_{j}\}$ will fetch a much better solution ($V_{j}$) because reducing the completion times by $t$ will lead to a situation where none of the jobs either start at time $0$ (because $ES>0$) nor any of the jobs finish at the due date since the jobs $1,2,3,\dots,j-1$ are early, jobs $j,j+1,\dots,n$ are tardy and the new completion time of job $j$ is $C_{j}^{'}=C_{j}-t$. 

Since after this reduction $DT_{j}>0$ and $DT_{j}<DT_{j+1}$ for $j=1,2,3,\dots,n-1$, none of the jobs will finish at the due date after a reduction by $t$ units. Moreover, it was proved by Cheng \emph{et al.}~\cite{cheng} that in an optimal schedule for the restrictive common due date, either one of the jobs should start at time $0$ or one of the jobs should end at the due date. This case can occur only if we reduce the completion times by $\min\{ES,DT_{j}\}$. If $ES<DT_{j}$ the first job will start at time $0$ and if $DT_j<ES$ then one of the jobs will end at the due date. In the next iterations we continue the reductions as long as we get an improvement in the solution and once the new solution is not better than the previous best then we do not need to check any further and we have our optimal solution. This can be proved by considering the values of the objective function at two iterations indices; $j$ and $j+1$. Let $V_{j}$ and $V_{j+1}$ be the value of the objective function at these two indexes then we can prove that the solution cannot be improved any further if $V_{j+1}>V_{j}$ by Lemma~\ref{lemma3}.
\end{proof}

\begin{lemma}\label{lemma3}
Once the value of the solution at any iteration $j$ is less than the value at iteration $j+1$, then the solution cannot be improved any further.

\end{lemma}
\begin{proof}
If $V_{j+1}>V_{j}$ then it means that further left shift of the jobs does not fetch a better solution. Note that the objective function has two parts, penalty due to earliness and penalty due to tardiness. Let us consider the earliness and tardiness of the jobs after the $j^{th}$ iterations are $E_{i}^{j}$ and $T_{i}^{j}$ for $i=1,2,\dots,n$. Then we have $V_{j}=\sum_{\substack{i=1}}^{n} (\alpha_{i}E_{i}^{j}+\beta_{i}T_{i}^{j})$ and $V^{j+1}=\sum_{\substack{i=1}}^{n} (\alpha_{i}E_{i}^{j+1}+\beta_{i}T_{i}^{j+1})$. Besides, after every iteration of the $while$ loop in Algorithm~\ref{main}, the completion times are reduced or in other words the jobs are shifted left. This leads to an increase in the earliness and a decrease in the tardiness of the jobs. Let's say, the difference in the reduction between $V^{j+1}$ and $V^{j}$ is $x$. Then we have $E^{j+1}=E^{j}+x$ and $T_{j+1}=T_{j}-x$. Since $V^{j+1}>V^{j}$, we have: $\sum_{\substack{i=1}}^{n} (\alpha_{i}E_{i}^{j+1}+\beta_{i}T_{i}^{j+1}) > \sum_{\substack{i=1}}^{n} (\alpha_{i}E_{i}^{j}+\beta_{i}T_{i}^{j})$. By substituting the values of $E^{j+1}$ and $T^{j+1}$ we get,  $\sum_{\substack{i=1}}^{j+1} \alpha_{i}x > \sum_{\substack{i=j+2}}^{n} \beta_{i}x$. Hence, at the $(j+1)^{th}$ iteration the total penalty due to earliness exceeds the total penalty due to tardiness. This proves that for any further reduction there can not be an improvement in the solution because a decrease in the tardiness penalty will always be less than the increase in the earliness penalty.
\end{proof}

\begin{figure}[ht]
\centering
\includegraphics[width=11cm,height=8.5cm]{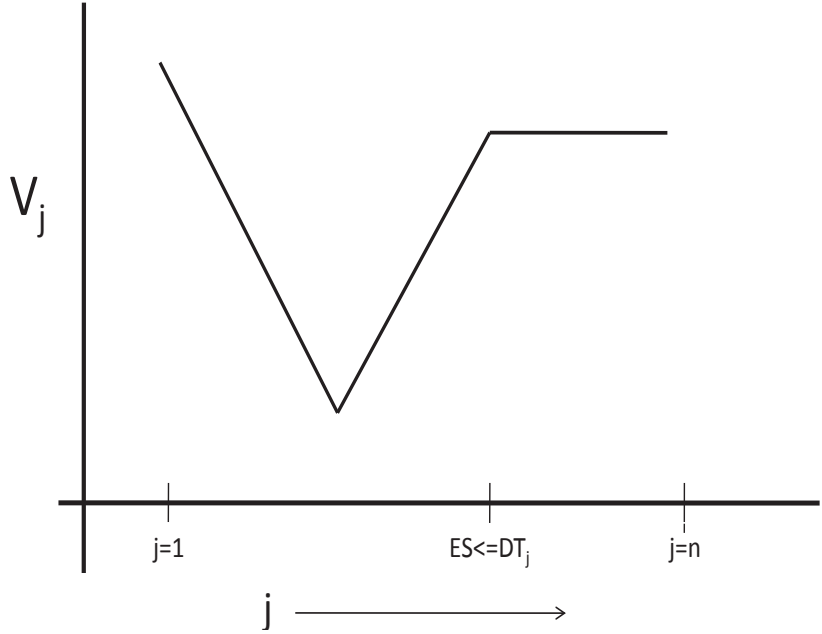}
\caption{The trend of the solution value against each iteration of Algorithm~\ref{main}, for a job sequence. The value of the solution does not improve any further after a certain number of reductions.}
\label{trend}
\end{figure}

\section{Algorithm Run-Time Complexity} 
In this section we study and prove the run-time complexity of the Algorithms~\ref{alotmac},~\ref{main} and~\ref{parallel}. We calculate the complexities of all the algorithms separately considering the worst cases for all. Let $T1, T2, T3$ be the run-time complexities of each algorithm respectively.

\begin{lemma}
The run-time complexity of Algorithm~\ref{parallel} is $O(mn^3)$ where $n$ and $m$ are the total number of jobs and machines, respectively.
\end{lemma}

\begin{proof}
To prove this we first consider the complexity of Algorithm~\ref{alotmac}. The most crucial step as per the run-time in Algorithm~\ref{alotmac} is the calculation of $\lambda$, which is of $O(m)$ order. The rest of the steps are of order $O(1)$. Hence, the overall complexity of Algorithm~\ref{alotmac} is $O(m+(n-m)m)$, where $O(m)$ corresponds to the first \textit{for} loop and $O((n-m)m)$ corresponds to the second \textit{for} loop involving the calculation for $\lambda$. Since the number of machines is usually much less than the number of jobs, we have $T1=O(mn)$.

As for Algorithm~\ref{main}, the calculations involved in the initialization step and evaluation of $PL,DT,ES,Sol$ are all of $O(n)$ complexity and their evaluation is irrespective of the any conditions unlike inside the $while$ loop. The $while$ loop again evaluates and updates these parameters at every step of its iteration and returns the output once their is no improvement possible. The worst case will occur when the $while$ loop is iterated over all the values of $j$, $j=2,3,\dots,n$. Hence the complexity of Algorithm~\ref{main} is $O(n^2)$ with $n$ being the number of jobs processed by the machine. Hence, $T2=O(n^2)$.

To calculate the complexity of Algorithm~\ref{parallel}, we need to consider all the cases of the number of jobs processed by each machine. Let $x_1,x_2,x_3,\dots,x_m$ be the number of jobs processed by the machines, respectively. Then, $\sum_{\substack{i=1}}^{m} x_i=n$. We make a reasonable assumption that the number of machines is less than the number of jobs, which is usually the case. In such a case the complexity of Algorithm~\ref{parallel} ($T3$) is equal to  $O(mn)\cdot \sum_{\substack{i=1}}^{m}O(x_{i}^{2})$. Since $\sum_{\substack{i=1}}^{m} x_{i} = n$, we have $\sum_{\substack{i=1}}^{m}O(x_{i}^{2})=O(n^2)$. Thus the complexity of Algorithm~\ref{parallel} is $O(mn^3)$.
\end{proof}

\section{Exponential Search: An efficient implementation of Algorithm~\ref{main}}
Algorithm~\ref{main} shifts the jobs to the left by reducing the completion times of all the jobs by $\min\{ES,DT_{j}\}$ on every iteration of the $while$ loop. The runtime complexity of the algorithm can be improved form $O(n^2)$ to $O(n\log n)$ by implementing an exponential search instead of a step by step reduction, as in Algorithm~\ref{main}. To explain this we first need to understand the slope of the objective function values for each iteration. In the proof of optimality of Algorithm~\ref{main}, we proved that there is only one minimum present in $V^{j}$ $\forall j$. Besides, the value of $DT_{j}$ increases for every $j$ as it depends on the completion times. Also note that the reduction in the completion times is made by $\min\{ES,DT_{j}\}$. Hence, if for any $j$, $ES \leq DT_{j}$ then every iteration after $j$ will fetch the same objective function value, $V^{j}$. Hence, the solution values after each iteration will have trend as shown in Figure~\ref{trend}.

With such a slope of the solution we can use the exponential search as opposed to a step by step search, which will in turn improve the run-time complexity of Algorithm~\ref{main}. This can be achieved by increasing or decreasing the step size of the $while$ loop by orders of $2$ ({\em i.e.\ \/}$2,2^2,2^3,\dots,n$) while keeping track of the slope of the solution. The index of the next iteration should be increased if the slope is negative and decreased if the slope is non-negative. At each step we need to keep track of the previous two indices and once the difference between the indices is less than the minimum of the two, then we need to perform binary search on the same lines. The optimum will be reached if both the adjacent solutions are greater than the current value. In this methodology we do not need to search for all values of $j$ but in steps of $2^j$. Hence the run-time complexity with exponential search will be $O(n\log n)$ for the single machine and $O(mn^2 \log n)$ for the parallel machine case.

\begin{table}
\centering
{\footnotesize 
\caption{Results obtained for single machine common due date problem and comparison with benchmark results provided in the OR Library~\cite{or}. For any given number of jobs there are $10$ different instances provided and each instance is designated a number $k$.}
\label{result}
\begin{tabular}{ |p{1cm}|p{1cm} p{0.9cm}|p{0.9cm} p{0.9cm}|p{0.9cm} p{0.9cm}|p{0.9cm} p{0.9cm}| }\hline
\textbf{Jobs} & \multicolumn{2}{|c|}{\textbf{h=0.2}} & \multicolumn{2}{|c|}{\textbf{h=0.4}} & \multicolumn{2}{|c|}{\textbf{h=0.6}} & \multicolumn{2}{|c|}{\textbf{h=0.8}} \\ \hline
\textbf{n=10} & \textbf{APSA} & \textbf{BR} & \textbf{APSA} & \textbf{BR} & \textbf{APSA} & \textbf{BR} & \textbf{APSA} & \textbf{BR} \\ \hline
\textbf{k=1} & 1936 & 1936 & 1025 & 1025 & 841 & 841 & 818 & 818\\ \hline
\textbf{k=2} & 1042 & 1042 & 615 & 615 & 615 & 615 & 615 & 615\\ \hline
\textbf{k=3} & 1586 & 1586 & 917 & 917 & 793 & 793 & 793 & 793\\ \hline
\textbf{k=4} & 2139 & 2139 & 1230 & 1230 & 815 & 815 & 803 & 803\\ \hline
\textbf{k=5} & 1187 & 1187 & 630 & 630 & 521 & 521 & 521 & 521\\ \hline
\textbf{k=6} & 1521 & 1521 & 908 & 908 & 755 & 755 & 755 & 755\\ \hline
\textbf{k=7} & 2170 & 2170 & 1374 & 1374 & 1101 & 1101 & 1083 & 1083\\ \hline
\textbf{k=8} & 1720 & 1720 & 1020 & 1020 & 610 & 610 & 540 & 540\\ \hline
\textbf{k=9} & 1574 & 1574 & 876 & 876 & 582 & 582 & 554 & 554\\ \hline
\textbf{k=10} & 1869 & 1869 & 1136 & 1136 & 710 & 710 & 671 & 671\\ \hline
\textbf{n=20} & \textbf{APSA} & \textbf{BR} & \textbf{APSA} & \textbf{BR} & \textbf{APSA} & \textbf{BR} & \textbf{APSA} & \textbf{BR} \\ \hline
\textbf{k=1} & 4394 & 4431 & 3066 & 3066 & 2986 & 2986 & 2986 & 2986\\ \hline
\textbf{k=2} & 8430 & 8567 & 4847 & 4897 & 3206 & 3260 & 2980 & 2980\\ \hline
\textbf{k=3} & 6210 & 6331 & 3838 & 3883 & 3583 & 3600 & 3583 & 3600\\ \hline
\textbf{k=4} & 9188 & 9478 & 5118 & 5122 & 3317 & 3336 & 3040 & 3040\\ \hline
\textbf{k=5} & 4215 & 4340 & 2495 & 2571 & 2173 & 2206 & 2173 & 2206\\ \hline
\textbf{k=6} & 6527 & 6766 & 3582 & 3601 & 3010 & 3016 & 3010 & 3016\\ \hline
\textbf{k=7} & 10455 & 11101 & 6279 & 6357 & 4126 & 4175 & 3878 & 3900\\ \hline
\textbf{k=8} & 3920 & 4203 & 2145 & 2151 & 1638 & 1638 & 1638 & 1638\\ \hline
\textbf{k=9} & 3465 & 3530 & 2096 & 2097 & 1965 & 1992 & 1965 & 1992\\ \hline
\textbf{k=10} & 4979 & 5545 & 3012 & 3192 & 2110 & 2116 & 1995 & 1995\\ \hline
\textbf{n=50} & \textbf{APSA} & \textbf{BR} & \textbf{APSA} & \textbf{BR} & \textbf{APSA} & \textbf{BR} & \textbf{APSA} & \textbf{BR} \\ \hline
\textbf{k=1} & 40936 & 42363 & 24146 & 24868 & 17970 & 17990 & 17982 & 17990\\ \hline
\textbf{k=2} & 31174 & 33637 & 18451 & 19279 & 14217 & 14231 & 14067 & 14132\\ \hline
\textbf{k=3} & 35552 & 37641 & 20996 & 21353 & 16497 & 16497 & \cellcolor{gray!50}16517 & \cellcolor{gray!50}16497\\ \hline
\textbf{k=4} & 28037 & 30166 & 17137 & 17495 & 14088 & 14105 & 14101 & 14105\\ \hline
\textbf{k=5} & 32347 & 32604 & 18049 & 18441 & 14615 & 14650 & 14615 & 14650\\ \hline
\textbf{k=6} & 35628 & 36920 & 20790 & 21497 & \cellcolor{gray!50}14328 & \cellcolor{gray!50}14251 & 14075 & 14075\\ \hline
\textbf{k=7} & 43203 & 44277 & 23076 & 23883 & 17715 & 17715 & 17699 & 17715\\ \hline
\textbf{k=8} & 43961 & 46065 & 25111 & 25402 & 21345 & 21367 & 21351 & 21367\\ \hline
\textbf{k=9} & 34600 & 36397 & 20302 & 21929 & 14202 & 14298 & \cellcolor{gray!50}14064 & \cellcolor{gray!50}13952\\ \hline
\textbf{k=10} & 33643 & 35797 & 19564 & 20048 & 14367 & 14377 & 14374 & 14377\\ \hline
\textbf{n=100} & \textbf{APSA} & \textbf{BR} & \textbf{APSA} & \textbf{BR} & \textbf{APSA} & \textbf{BR} & \textbf{APSA} & \textbf{BR} \\ \hline
\textbf{k=1} & 148316 & 156103 & 89537 & 89588 & 72017 & 72019 & 72017 & 72019\\ \hline
\textbf{k=2} & 129379 & 132605 & 73828 & 74854 & 59350 & 59351 & 59348 & 59351\\ \hline
\textbf{k=3} & 136385 & 137463 & 83963 & 85363 & \cellcolor{gray!50}68671 & \cellcolor{gray!50}68537 & \cellcolor{gray!50}68670 & \cellcolor{gray!50}68537\\ \hline
\textbf{k=4} & 134338 & 137265 & 87255 & 87730 & 69192 & 69231 & 69039 & 69231\\ \hline
\textbf{k=5} & 129057 & 136761 & 74626 & 76424 & 55291 & 55291 & 55275 & 55277\\ \hline
\textbf{k=6} & 145927 & 151938 & 81182 & 86724 & 62507 & 62519 & 62410 & 62519\\ \hline
\textbf{k=7} & 138574 & 141613 & 79482 & 79854 & \cellcolor{gray!50}62302 & \cellcolor{gray!50}62213 & 62208 & 62213\\ \hline
\textbf{k=8} & 164281 & 168086 & 95197 & 95361 & 80722 & 80844 & 80841 & 80844\\ \hline
\textbf{k=9} & 121189 & 125153 & 72817 & 73605 & 58769 & 58771 & 58771 & 58771\\ \hline
\textbf{k=10} & 121425 & 124446 & \cellcolor{gray!50}72741 & \cellcolor{gray!50}72399 & 61416 & 61419 & 61416 & 61419\\ \hline
\end{tabular}
}
\end{table}
\begin{table}
\centering
{\footnotesize
\caption{Results obtained for single machine common due date problem and comparison with benchmark results provided in the OR Library~\cite{or}. There are $10$ different instances provided and each instance is designated a number $k$.}
\label{result1}
\begin{tabular}{ |p{1.05cm}|p{0.9cm} p{0.9cm}|p{0.9cm} p{0.9cm}|p{0.9cm} p{0.9cm}|p{0.9cm} p{0.9cm}| }\hline
\textbf{Jobs} & \multicolumn{2}{|c|}{\textbf{h=0.2}} & \multicolumn{2}{|c|}{\textbf{h=0.4}} & \multicolumn{2}{|c|}{\textbf{h=0.6}} & \multicolumn{2}{|c|}{\textbf{h=0.8}} \\ \hline
\textbf{n=200} & \textbf{APSA} & \textbf{BR} & \textbf{APSA} & \textbf{BR} & \textbf{APSA} & \textbf{BR} & \textbf{APSA} & \textbf{BR} \\ \hline
\textbf{k=1} & 523042 & 526666 & 300079 & 301449 & 254268 & 254268 & \cellcolor{gray!50}254362 & \cellcolor{gray!50}254268\\ \hline
\textbf{k=2} & 557884 & 566643 & 333930 & 335714 & \cellcolor{gray!50}266105 & \cellcolor{gray!50}266028 & \cellcolor{gray!50}266549 & \cellcolor{gray!50}266028\\ \hline
\textbf{k=3} & 510959 & 529919 & 303924 & 308278 & 254647 & 254647 & 254572 & 254647\\ \hline
\textbf{k=4} & 596719 & 603709 & 359966 & 360852 & \cellcolor{gray!50}297305 & \cellcolor{gray!50}297269 & \cellcolor{gray!50}297729 & \cellcolor{gray!50}297269\\ \hline
\textbf{k=5} & 543709 & 547953 & 317707 & 322268 & \cellcolor{gray!50}260703 & \cellcolor{gray!50}260455 & 260423 & 260455\\ \hline
\textbf{k=6} & 500354 & 502276 & 287916 & 292453 & 235947 & 236160 & 236013 & 236160\\ \hline
\textbf{k=7} & 477734 & 479651 & 279487 & 279576 & 246910 & 247555 & 247521 & 247555\\ \hline
\textbf{k=8} & 522470 & 530896 & 287932 & 288746 & 225519 & 225572 & \cellcolor{gray!50}225897 & \cellcolor{gray!50}225572\\ \hline
\textbf{k=9} & 561956 & 575353 & 324475 & 331107 & 254953 & 255029 & 254956 & 255029\\ \hline
\textbf{k=10} & 560632 & 572866 & 328964 & 332808 & 269172 & 269236 & 269208 & 269236\\ \hline
\end{tabular}
}
\centering
{\footnotesize
\caption{Average run-times in seconds for the single machine cases for the obtained solutions. The average run-time for any job is the average of all the $40$ instances.}
\label{runtimeS}
\begin{tabular}{ |c|p{0.02cm}cp{0.02cm}|p{0.02cm}cp{0.02cm}|p{0.02cm}cp{0.02cm}|p{0.02cm}cp{0.02cm}|p{0.02cm}cp{0.02cm}|}\hline
\textbf{No. of Jobs} && 10 &&& 20 &&& 50 &&& 100 &&& 200 &\\ \hline
\textbf{BR} && 0.9 &&& 47.8 &&& 87.3 &&& 284.9 &&& 955.2& \\ \hline
\textbf{APSA} && 0.46 &&& 1.12 &&& 22.17 &&& 55.22 &&& 132.32 & \\ \hline
\end{tabular}
}
\end{table}

\section{Results}
In this section we present our results for the single and parallel machine cases. We used our exact algorithms with simulated annealing for finding the best job sequence. All the algorithms were implemented on MATLAB and run on a machine with a $1.73$ GHz processor and $2$ GB RAM. We first present our results for the benchmark instances provided by Biskup~\emph{et al.} in~\cite{biskup} and compare our results with the results provided in the OR-library~\cite{or}, for the single machine case. 

We use a modified Simulated Annealing algorithm to generate job sequences and Algorithm~\ref{main} to optimize each sequence to its minimum penalty. The ensemble size for SA is taken to be $20$ for all the instances. The initial temperature is kept as twice the standard deviation of the energy at infinite temperature: $\sigma_{E_{T=\infty}} = \sqrt{\langle E^2 \rangle_{T=\infty} - \langle E \rangle^2_{T=\infty}}$. We estimate this quantity by randomly sampling the configuration space~\cite{salamon}. An exponential schedule for cooling is adopted with a cooling rate of $0.999$. One of the modifications from the standard SA is in the acceptance criterion. We implement two acceptance criteria: the Metropolis acceptance probability, $\min\{1,\exp((-\hspace{-0.3em}\bigtriangleup\hspace{-0.3em} E)/T)\}$~\cite{salamon} and a constant acceptance probability of $0.07$. A solution is accepted with this constant probability if it is rejected by the Metropolis criterion. This concept of a constant probability is useful when the SA is run for many iterations and the metropolis acceptance probability is almost zero, since the temperature would become infinitesimally small. Apart from this, we also incorporate elitism in our modified SA. Elitism has been successfully adopted in evolutionary algorithms for several complex optimization problems~\cite{elitist1,elitist2}. We observed that this concept works well for the CDD problem. As for the perturbation rule, we first randomly select a certain number of jobs in any job sequence and permute them randomly to create a new sequence. The number of jobs selected for this permutation is taken as  $2 + \lfloor\sqrt{n/10}\rfloor$, where $n$ is the number of jobs. For large instances the size of this permutation is quite small but we have observed that it works well with our modified simulated annealing algorithm.

\begin{table}
\centering
{\footnotesize
\caption{Results obtained for parallel machines for the benchmark instances for $k=1$ with $2$, $3$ and $4$ machines up to $200$ jobs.}
\label{result3}
\begin{tabular}{ | c | c | c | c |c|}\hline
\textbf{No. of Jobs} & \textbf{Machines} & \textbf{h value} &  \textbf{Results Obtained} & \textbf{Run-Time (seconds)}  \\ \hline
\multirow{6}{*}{\textbf{10}} & \multirow{2}{*}{2} & 0.4 & 612 & 0.0473 \\ \cline{3-5}
& & 0.8 & 398 & 0.0352 \\ \cline{2-5}
& \multirow{2}{*}{3} &  0.4 &  507 & 0.0239 \\ \cline{3-5}
& & 0.8 & 256 & 0.0252 \\ \cline{2-5}
& \multirow{2}{*}{4} & 0.4 & 364 & 0.0098 \\ \cline{3-5}
& & 0.8 &  197 & 0.0157 \\ \hline
\multirow{6}{*}{\textbf{20}} & \multirow{2}{*}{2} & 0.4 & 1527 & 0.4061 \\ \cline{3-5}
& & 0.8 & 1469 & 0.6082 \\ \cline{2-5}
& \multirow{2}{*}{3} &  0.4 &  1085 & 3.4794 \\ \cline{3-5}
& & 0.8 & 957 & 7.8108 \\ \cline{2-5}
& \multirow{2}{*}{4} & 0.4 & 848 & 8.5814 \\ \cline{3-5}
& & 0.8 &  686 & 8.4581 \\ \hline
\multirow{6}{*}{\textbf{50}} & \multirow{2}{*}{2} & 0.4 & 12911 & 7.780 \\ \cline{3-5}
& & 0.8 & 9020 & 55.3845 \\ \cline{2-5}
& \multirow{2}{*}{3} &  0.4 &  8913 & 59.992 \\ \cline{3-5}
& & 0.8 & 6010 & 125.867 \\ \cline{2-5}
& \multirow{2}{*}{4} & 0.4 & 7097 & 153.566 \\ \cline{3-5}
& & 0.8 &  4551 & 22.347 \\ \hline
\multirow{6}{*}{\textbf{100}} & \multirow{2}{*}{2} & 0.4 & 45451 & 101.475 \\ \cline{3-5}
& & 0.8 & 37195 & 147.832 \\ \cline{2-5}
& \multirow{2}{*}{3} &  0.4 &  31133 & 159.872 \\ \cline{3-5}
& & 0.8 & 25097 & 186.762 \\ \cline{2-5}
& \multirow{2}{*}{4} & 0.4 & 23904 & 236.132 \\ \cline{3-5}
& & 0.8 &  19001 & 392.967 \\ \hline
\multirow{6}{*}{\textbf{200}} & \multirow{2}{*}{2} & 0.4 & 154094 & 165.436 \\ \cline{3-5}
& & 0.8 & 133848 & 231.768 \\ \cline{2-5}
& \multirow{2}{*}{3} &  0.4 &  103450 & 226.140 \\ \cline{3-5}
& & 0.8 & 96649 & 365.982 \\ \cline{2-5}
& \multirow{2}{*}{4} & 0.4 & 81437 & 438.272 \\ \cline{3-5}
& & 0.8 &  71263 & 500.00 \\ \hline
\end{tabular}
}
\end{table}
In Table~\ref{result} and~\ref{result1} we present our results (APSA) for the single machine case and compare them with the benchmark results (BR) provided in the OR-library~\cite{or}. The first $40$ instances with $10$ jobs each have been already solved optimally by Biskup~\emph{et al.} and we reach the optimality for all these instances within an average run-time of $0.457$ seconds. 
Among the next $160$ instances we achieve equal results for $13$ instances, better results for $133$ instances and for the remaining $14$ instances with $50$, $100$ and $200$ jobs, our results are within a gap of $0.803$ percent, $0.1955$ percent and $0.1958$ percent respectively. Feldmann and Biskup~\cite{feldmann} solved these instances using three metaheuristic approaches, namely: simulated annealing, evolutionary strategies and threshold accepting; and presented the average run-time for the instances on a Pentium/$90$ PC. 
In Table~\ref{runtimeS} we show our average run-times for the instances and compare them with the heuristics approach considered in~\cite{feldmann}. Apparently our approach is faster and achieves better results. However, there is a difference in the machines used for the implementation of the algorithms. In Table~\ref{result3} we present results for the same problem but with parallel machines for the Biskup benchmark instances. The computation has been carried out for $k=1$ up to $200$ jobs and a different number of machines with restrictive factor $h$. We make a change in the due date as the number of machines increases and assume that the due date $D$ is $D=\lfloor h \cdot \sum_{\substack{i=1}}^{n} P_{i}/m\rfloor$. This assumption makes sense as an increase in the number of machines means that the jobs can be completed much faster and reducing the due-date will test the whole setup for more competitive scenarios. We implemented Algorithm~\ref{parallel} with six different combinations of the number of machines and the restrictive factor. Since these instances have not been solved for the parallel machines, we are presenting the upper bounds achieved for these instances using Algorithm~\ref{parallel} and the modified simulated annealing.

\section{Conclusion and Future Direction}
In this paper we present two novel exact polynomial algorithms for the common due-date problem to optimize any given job sequence. We prove the optimality for the single machine and the run-time complexity of the algorithms. We implemented our algorithms over the benchmark instances provided by Biskup~\emph{et al.}~\cite{biskup} up to $200$ jobs. The results obtained by using our algorithms are superior to the benchmark results in quality. We also present results for the parallel machine case for the same instances.

In future we intend to study our algorithms over all the instances up to $1000$ jobs with some hybridized metaheuristic approach. We are planning to implement these metaheuristic approaches using graphics processing units (GPUs) and provide speed-ups in runtime for all the instances. Further we plan to implement it over all $280$ benchmark instances for several independent machines.

\section*{Acknowledgement}
The research project was promoted and funded by the European Union and the Free State of Saxony, Germany. The authors take the responsibility for the content of this publication.

\end{document}